\pdfoutput=1
\documentclass[]{article}  
\usepackage{url,float}
\usepackage{graphicx}
\usepackage{amsmath}
\usepackage{amsfonts}
\usepackage{amssymb}
\usepackage{latexsym}
\usepackage{mathtools}
\usepackage{color}

\newcommand{\hide}[1]{}

\newcommand{\ABox}{
\raisebox{3pt}{\framebox[6pt]{\rule{6pt}{0pt}}}
}
\newenvironment{proof}{{\bf Proof:}}{\hfill\ABox}

\newtheorem{theorem}{{\bf Theorem}}

\newtheorem{lemma}{Lemma}
\newcommand{\lemlab}[1]{\label{lemma:#1}}
\newcommand{\thmlab}[1]{\label{thm:#1}}

\newcommand{\figlab}[1]{\label{fig:#1}}
\newcommand{\seclab}[1]{\label{sec:#1}}

\newcommand{\lemref}[1]{\ref{lemma:#1}}
\newcommand{\thmref}[1]{\ref{thm:#1}}

\newcommand{\figref}[1]{\ref{fig:#1}}

\def\q{{\theta}}
\newcommand{\squeezelist}{\setlength{\itemsep}{0pt}}
\usepackage{enumitem}
\title{%
Hamiltonian Quasigeodesics yield Nets
} %title

\author{%
Joseph O'Rourke%
\thanks{Smith College, Northampton, MA, USA. \texttt{jorourke@smith.edu}.}
}

\date{\today}

\begin{document}
\maketitle

\begin{abstract}
This note establishes that every %convex 
polyhedron
that has a Hamiltonian quasigeodesic can be edge-unfolded to a net,
and shows that the class of such polyhedra is infinite.
\end{abstract}

\section{Introduction}
%\cite{QuasiTwist}
%\cite{p-qglcs-49}
%\cite{lddss-zupc-10}
This note establishes one result (Theorem~\thmref{HamQnet}) that is more an
observation than a theorem, as it largely depends on definitions and a straightforward argument.
Nevertheless, it may be of some interest, making connections between several different aspects of
convex polyhedra.

We start with the definitions needed to describe the result.
We restrict attention to convex polyhedra.
An \emph{edge-unfolding} of a convex polyhedron $P$ is a collection of edge cuts forming a spanning tree
of the vertices,
which unfolds the surface of $P$ to one piece in the plane.
If the planar unfolding is a (weakly) simple polygon, with no self-overlap, it is known as a \emph{net}.
It is a long unsolved problem to decide whether or not every convex polyhedron has a net.
This has become known as ``D\"urer's problem"~\cite{do-gfalop-07}~\cite{o-dp-13}.
The conclusion of Theorem~\thmref{HamQnet} is that, under certain conditions, $P$ has a net,
and there is an infinite class of such $P$.

A \emph{quasigeodesic} has at most $\pi$ surface angle to each side at every point, 
in contrast to geodesics which have exactly $\pi$ to each side.
Quasigeodesics can pass through vertices, and are geodesic segments between vertices.
Pogorelov proved that every convex polyhedron 
has at least three simple closed quasigeodesics~\cite{p-qglcs-49}. 
Here we focus on quasigeodesics that follow edges of the $1$-skeleton of $P$.
A \emph{Hamiltonian quasigeodesic} is a simple closed quasigeodesic following edges of $P$,
and passing through every vertex of $P$. Thus a Hamiltonian quasigeodesic is
a Hamiltonian circuit, but with angle restrictions at each vertex.
This notion was introduced and played a role in~\cite{QuasiTwist}.
Earlier, Hamiltonian unfoldings without the quasigeodesic condition (also called \emph{zipper edge-unfoldings}) were studied in~\cite{lddss-zupc-10}.

Of course not every convex polyhedron has a Hamiltonian quasigeodesic because some polyhedra
have no Hamiltonian circuit, for example, the rhombic dodecahedron.
In~\cite{QuasiTwist} it was shown that three of the Platonic solids---tetrahedron, octahedron, cube---have
a Hamiltonian quasigeodesic, but the dodecahedron and the icosahedron do not.
Note that the boundary of a doubly-covered convex polygon, which is treated as a polyhedron in this literature,
is a Hamiltonian quasigeodesic.

We can now state the theorem.
\begin{theorem}
\thmlab{HamQnet}
If a convex polyhedron $P$ has a Hamiltonian quasigeodesic $Q$,
then there is an edge-unfolding of $P$ to a net.
\end{theorem}
\begin{proof}
The proof follows from two straightforward claims.
\begin{enumerate}[label={(\arabic*)}]
\item The portion of $P$ enclosed to the right or left of $Q$ is isometric to a convex polygon.
\item Joining two convex polygons $A$ and $B$ along a shared edge $e$ avoids overlap 
between $A$ and $B$.
\end{enumerate}
We now add more detail to these claims. Consider $Q$ directed, partitioning $P$
into two 
``halves" $A$ and $B$, with $P = A \cup B$.
\begin{enumerate}[label={(\arabic*)}]
\item Because $Q$ is a quasigeosdic, the angle to the left of every vertex of $Q$ is $\le \pi$.
Because $Q$ passes through every vertex of $P$, there are no vertices of $P$ to the left,
enclosed by $Q$, and 
because all curvature is concentrated at vertices, no curvature.
Therefore the region $A$ of the surface of $P$ to the left of $Q$ is isometric to
a planar convex polygon (convex because of the $\le \pi$ condition); and similarly for $B$.
Let $\bar{A}$ and $\bar{B}$ be planar embeddings of $A$ and $B$.
\item Although $\bar{A}$ and $\bar{B}$ are not necessarily congruent, their
boundaries are each composed of the same edges of $Q$.
Select any edge $e \in Q$. Then joining $\bar{A}$ to $\bar{B}$ to either side of and sharing $e$
produces a non-overlapping simple polygon---a net---because they sit on opposite sides
of the line containing $e$, and so cannot overlap one another.
One can view this as unfolding after cutting all edges of $Q$ but $e$, a spanning cut-path.
\end{enumerate}
\end{proof}

\section{Three Examples}
We present three examples. In each we label faces as F, R, K, L, T, B,
for Front, Right, bacK, Left, Top, Bottom.
The first is a pyramid, the top half of a regular octahedron:
see Fig.~\figref{PyrLayout_2}.
Note the angles along $Q$ are $\pi \times \{ \frac{1}{3}, \frac{1}{2}, \frac{2}{3}, \frac{5}{6} , 1\}$---all $\le \pi$.
Two different ways of joining $A$ to $B$ are shown.
%%%%%%%%%%%%%%%%%%%%%%%%%%%%%%%%%Figure Begin
\begin{figure}[htbp]
\centering
\includegraphics[height=0.95\textheight]{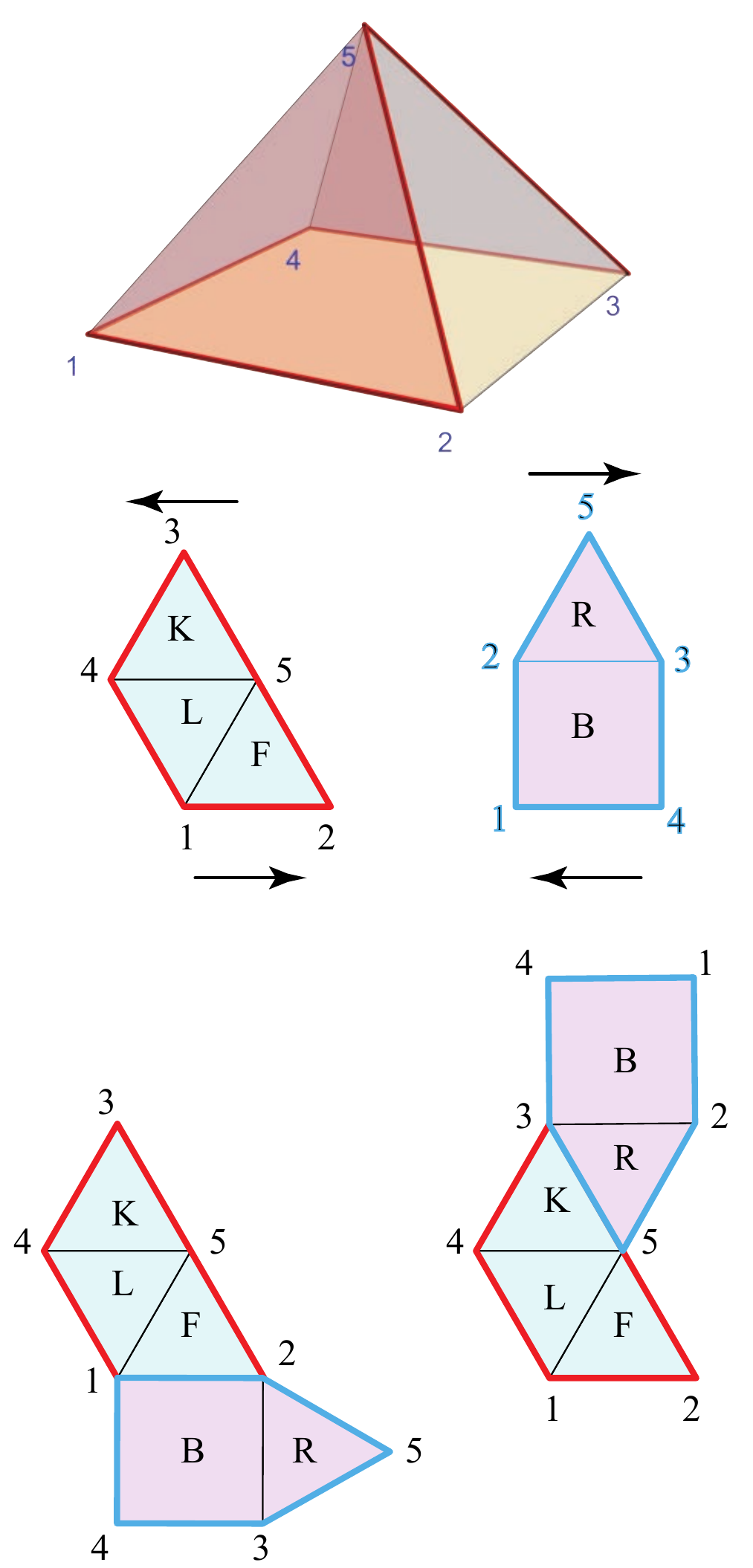}
%\fbox{xxx}
\caption{$Q=12534$. Arrows indicate counterclockwise (red) ordering around $A$,
and clockwise (blue) ordering around $B$.}
\figlab{PyrLayout_2}
\end{figure}
%%%%%%%%%%%%%%%%%%%%%%%%%%%%%%%%%Figure End

Our second example is the full regular octahedron:
see Fig.~\figref{OctQLayout_3}.
Here and in the next example, $\bar{A}$ and $\bar{B}$ are congruent.
Again two unfoldings are shown.
%%%%%%%%%%%%%%%%%%%%%%%%%%%%%%%%%Figure Begin
\begin{figure}[htbp]
\centering
\includegraphics[height=0.95\textheight]{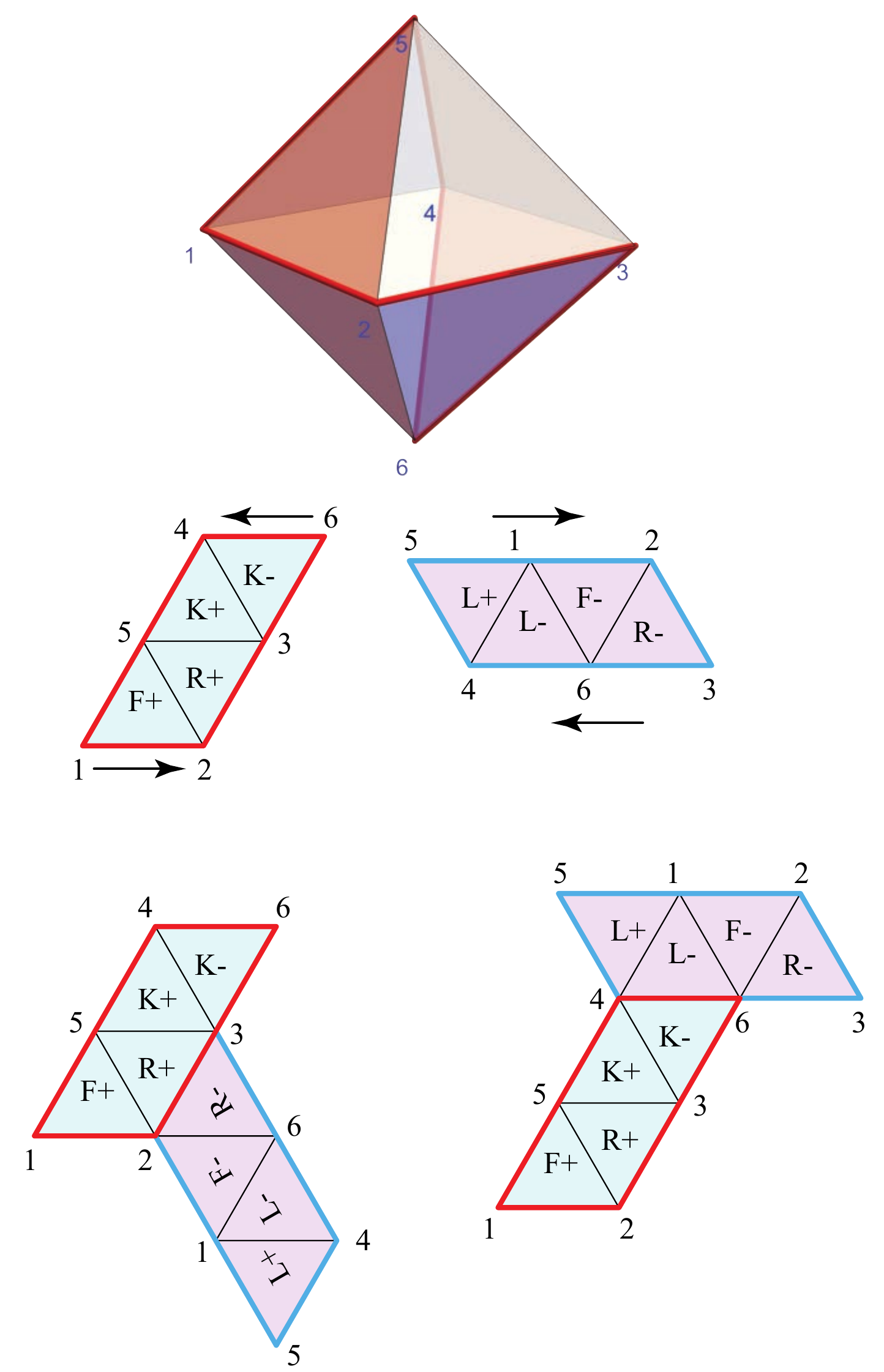}
%\fbox{xxx}
\caption{$Q=123645$.}
\figlab{OctQLayout_3}
\end{figure}
%%%%%%%%%%%%%%%%%%%%%%%%%%%%%%%%%Figure End

Finally, 
Fig.~\figref{CubeLayout_5} illustrates the cube with $Q$ forming a ``napkin holder,"
leading to two of the $11$ nets of a cube.
%%%%%%%%%%%%%%%%%%%%%%%%%%%%%%%%%Figure Begin
\begin{figure}[htbp]
\centering
\includegraphics[height=0.95\textheight]{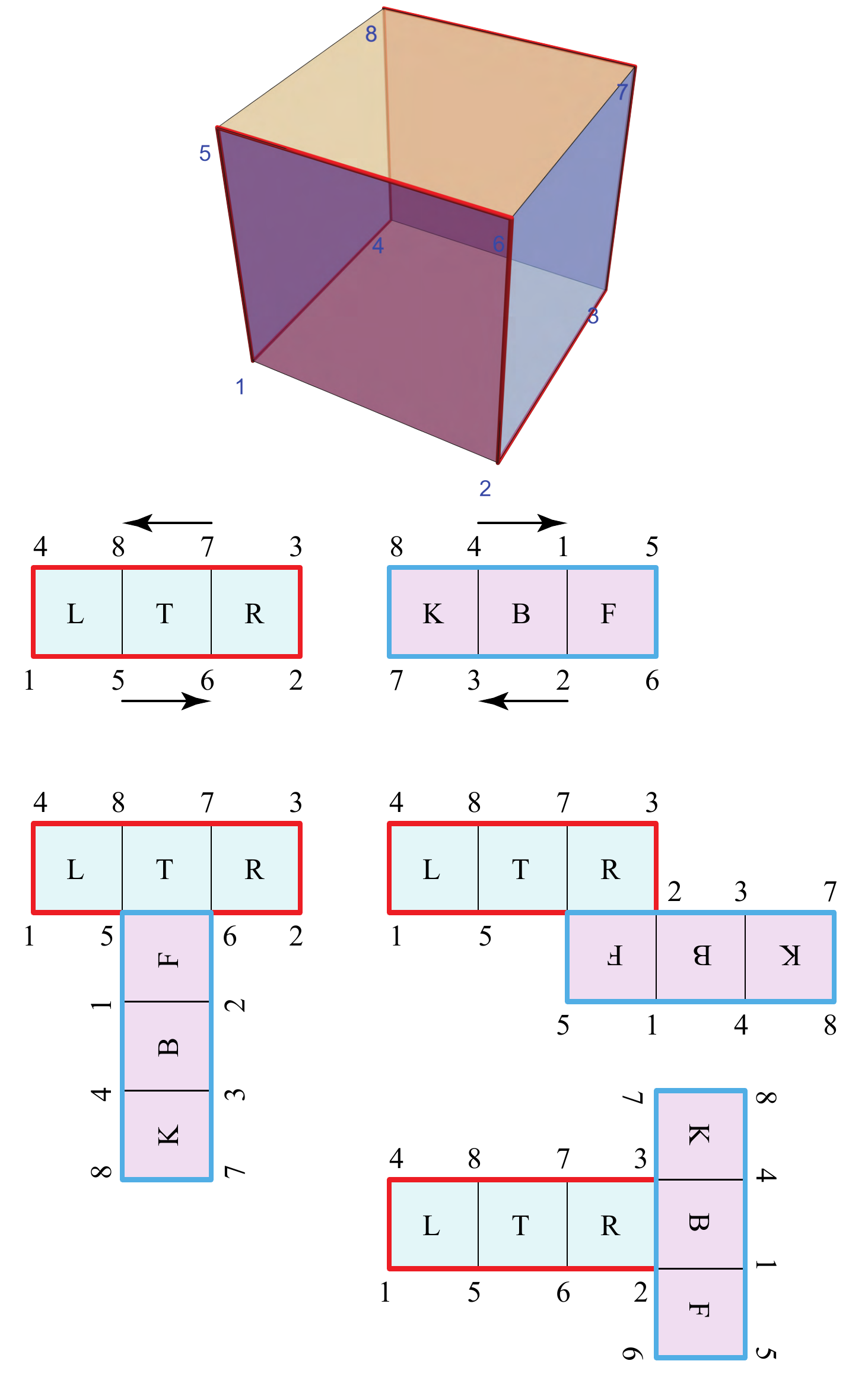}
%\fbox{xxx}
\caption{$Q=15623784$.}
\figlab{CubeLayout_5}
\end{figure}
%%%%%%%%%%%%%%%%%%%%%%%%%%%%%%%%%Figure End

\section{Infinite Class of Polyhedra: Antiprisms}
\seclab{Infinite}
Because D\"urer's problem is long unsolved, it is of some interest to identify
infinite classes of polyhedra that have an edge-unfolding to a net.
E.g., see~\cite[Sec.~22.5]{do-gfalop-07}.
Here we argue that a class of antiprisms satisfy Theorem~\thmref{HamQnet},
but does not advance on D\"urer's problem.

An \emph{antiprism} is a subclass of prismatoids: it is the convex hull
of top and bottom
congruent regular $n$-gons $A,B$ in horizontal parallel planes.
A special case is when the top is rotated $\pi/n$ about a centered vertical axis,
then known as a right antiprism.
Let $h$ be the height of an antiprism and $\q$ the rotation of its top.

\begin{lemma}
\lemlab{antiprism}
Every antiprism with height $h \le h^*(n,\q)$,
where $h^*(n,\q)$ is a positive constant dependent on $n$ and $\q$,
has a Hamiltonian quasigeodesic formed by the edges of
the band separating $A$ and $B$.
\end{lemma}
\begin{proof}
It is intuitive clear that as $h \to 0$, the zigzag path of edges of the
band becomes closer and closer to a flat path, and so turning gently.
Fig.~\figref{AntiPrism_ab} illustrates two $n=4$ examples with different
$\q$, both forming turn angles of $\pi$.
Elementary calculations determine $h^*(n,\q)$,
such that all smaller values of $h$ lead to the zigzag path forming a quasigeodesic. 
\end{proof}

\medskip
Although this establishes a continuum of polyhedra with Hamiltonian
quasigeodesics, these polyhedra each have a straightforward ``band'' edge unfolding~\cite{o-ufncp-13}.
This is because the band of lateral triangles are all congruent, and they
unfold to a straight strip. It is then easy to place $A$ and $B$ on opposite
sides of the band to form a net.
%\footnote{
%Less obviously, this unfolding to a net works as well with twisted nonconvex antiprisms.}

%see Fig.~\figref{AntiPrism_ab}.
%%%%%%%%%%%%%%%%%%%%%%%%%%%%%%%%%Figure Begin
\begin{figure}[htbp]
\centering
\includegraphics[width=0.75\linewidth]{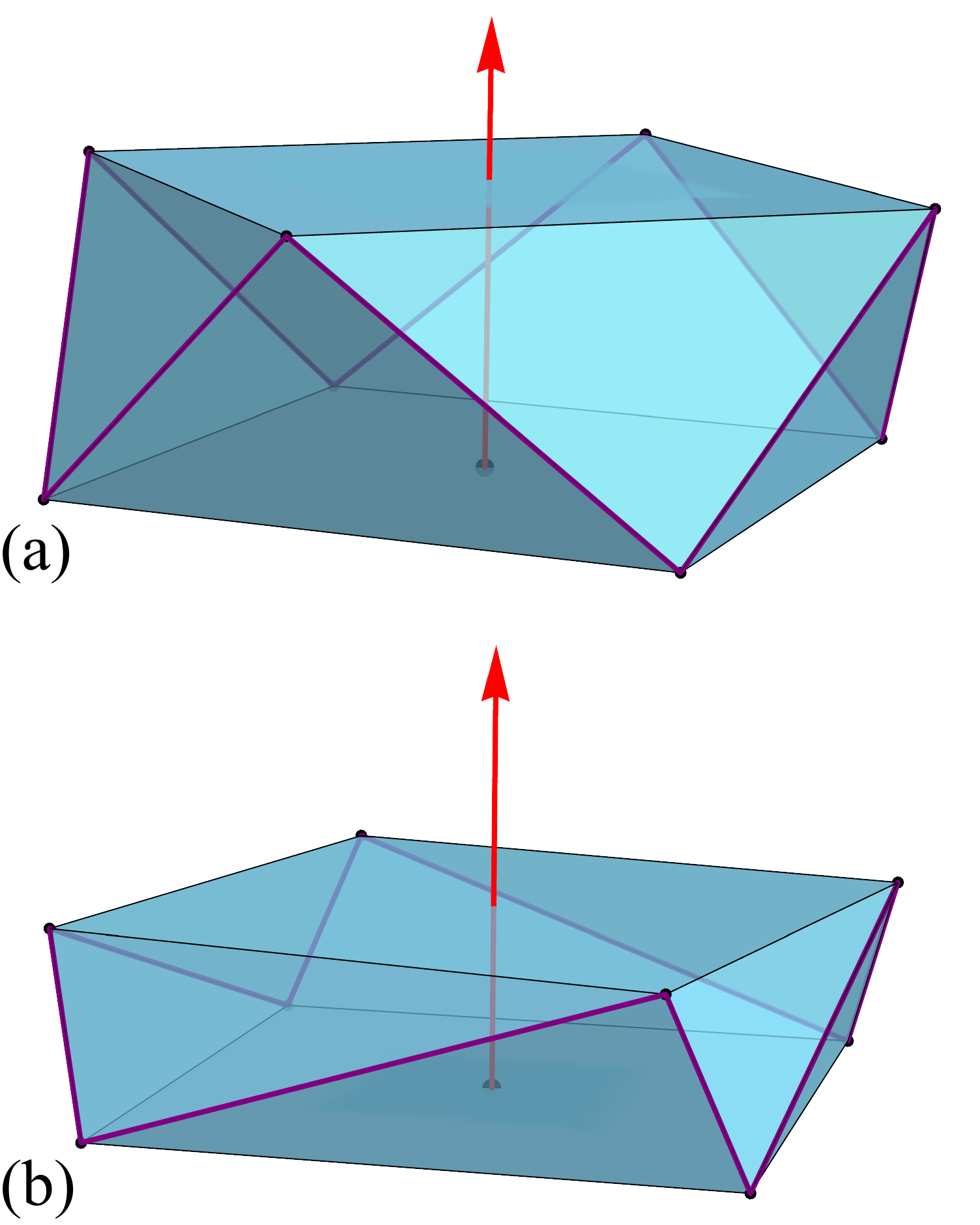}
\caption{Quasigeodesic turns are each $\pm \pi$. 
(a)~Twist angle $\q=\pi/4$, $h \approx 0.64$.
(b)~Twist angle $\q=0.2$, $h \approx 0.42$.}
\figlab{AntiPrism_ab}
\end{figure}
%%%%%%%%%%%%%%%%%%%%%%%%%%%%%%%%%Figure End

\section{Open Problems}
Although simple closed quasigeodecs are always present, it appears that
Hamiltonian quasigeodesics are relatively rare. It would be useful 
(and likely difficult) to characterize those convex
polyhedra that have a Hamiltonian quasigeodesic, for we then know that each has a net.

%It would be of interest to augment the 
%specific countably infinite regular polygon gluings described in
%Section~\secref{Infinite}, to obtain a continuum of examples satisfying Theorem~\thmref{HamQnet}.

I have not explored non-convex polyhedra $P$, but Theorem~\thmref{HamQnet} still holds as long as $P$ has a Hamiltonian quasigeodesic.
For example, twisted nonconvex antiprisms also satisfy 
a version of Lemma~\lemref{antiprism}:
see Fig.~\figref{TwistedAntiPrism}.
But they also have a band edge unfolding.
%%%%%%%%%%%%%%%%%%%%%%%%%%%%%%%%%Figure Begin
\begin{figure}[htbp]
\centering
\includegraphics[width=0.75\linewidth]{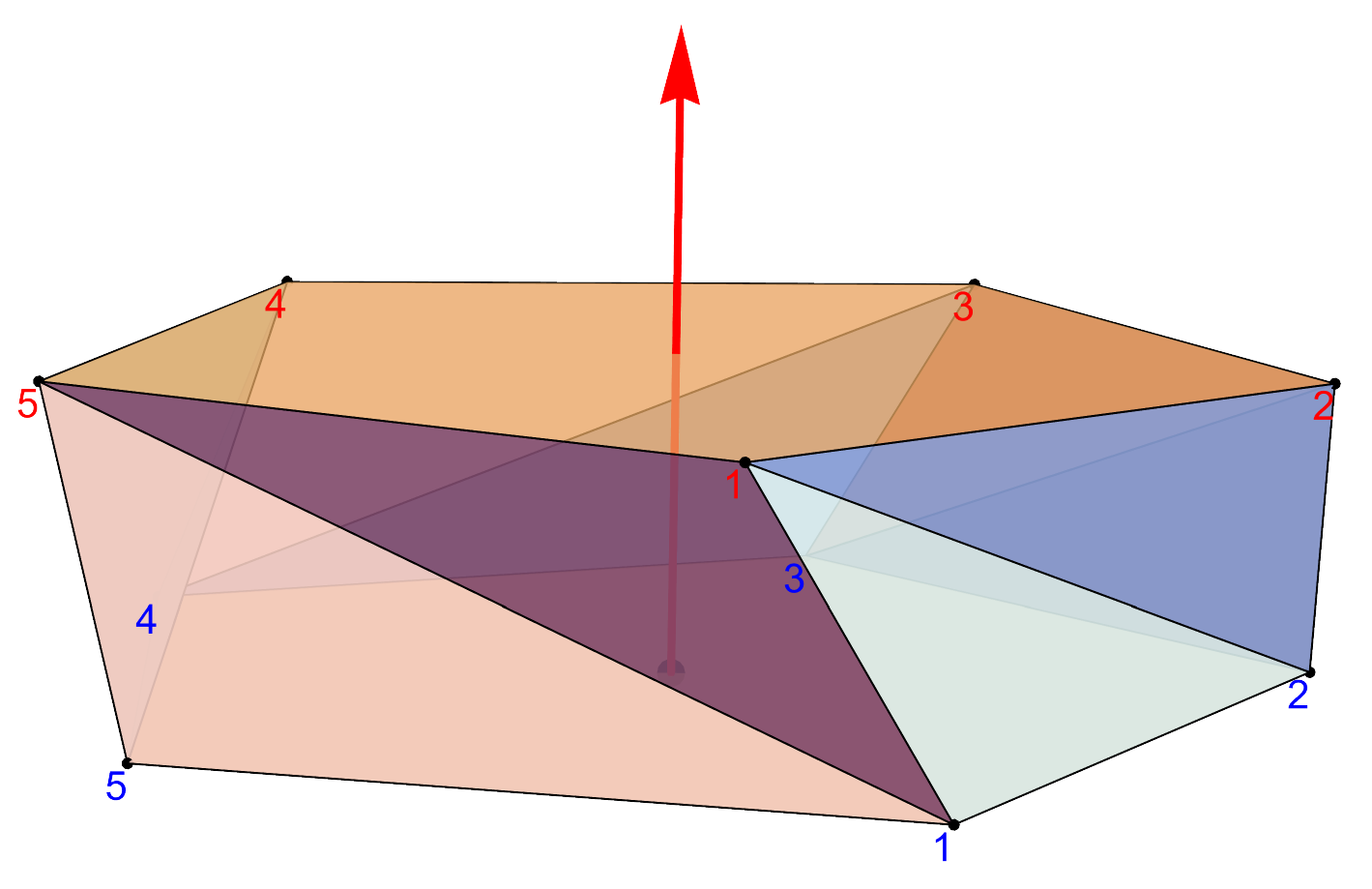}
\caption{A twisted nonconvex pentagonal pyramid.}
\figlab{TwistedAntiPrism}
\end{figure}
%%%%%%%%%%%%%%%%%%%%%%%%%%%%%%%%%Figure End

\paragraph{Acknowledgements.}
I thank my coauthors on~\cite{QuasiTwist} for stimulating discussions.

\bibliographystyle{alpha}
%\bibliography{/Users/jorourke/Documents/geom} 
\newcommand{\etalchar}[1]{$^{#1}$}

\end{document}